%% file: main.tex
\documentclass{llncs}

\title{Online Self-Indexed Grammar Compression\thanks{This work was supported by JSPS KAKENHI(24700140,26280088) and the JST PRESTO program}}

\author{Yoshimasa Takabatake\inst{1} 
  \and Yasuo Tabei\inst{2}
  \and Hiroshi Sakamoto\inst{1}}
\institute{
  Kyushu Institute of Technology~\email{\{takabatake,hiroshi\}@donald.ai.kyutech.ac.jp}
  \and 
  PRESTO, Japan Science and Technology Agency~\email{tabei.y.aa@m.titech.ac.jp}
}

\usepackage{makeidx}
\usepackage{algorithm}
\usepackage{algorithmicx}
\usepackage{algpseudocode}
\usepackage{amsmath}
\usepackage{url}
\usepackage[dvips]{graphicx}
\usepackage{comment}

\newcommand{\logstar}{\lg^*\hspace{-.9mm}}

\begin{document}
\maketitle
\begin{abstract}
Although several grammar-based self-indexes have been proposed thus far, 
their applicability is limited to offline settings where whole input texts are prepared, 
thus requiring to rebuild index structures for given additional inputs, which is often the case 
in the big data era. 
In this paper, we present the first online self-indexed grammar compression named 
OESP-index that can gradually build the index structure by reading input characters one-by-one.
Such a property is another advantage which enables saving a working space for construction, because 
we do not need to store input texts in memory. 
We experimentally test OESP-index on the ability to build index structures 
and search query texts, and we show OESP-index's efficiency, especially space-efficiency for 
building index structures.

\end{abstract}

\input{sec1} 
\input{sec2} 
\input{sec3} 
\input{sec4} 
\input{sec5}

\bibliographystyle{plain}
\bibliography{biblio}

\end{document}

%% file: sec1.tex
\section{Introduction}
Text collections including many repetitions, so called highly repetitive texts, have become common.
Version controlled software stores a large amount of documents with small differences. 
The current sequencing technology enables us to read individual genomes quickly and economically, which
generates large databases of thousands of human genomes~\cite{1000genomes}.
The genetic difference between individual human genomes is said to be approximately 0.1 percent, thus making the collection highly repetitive. 
There is therefore a strong need for developing powerful methods to store and process repetitive text collections on a large-scale. 

Self-indexes aim at representing a collection of texts in a compressed format that supports the random access to any position and also provides query searches on the collection. 
Although grammar-based self-indexes are especially effective for processing highly repetitive texts 
and several grammar-based self-indexes have been proposed~\cite{Gagie2012,Gagie14,Claude2010,Claude12,Navarro08} (See Table~\ref{tab:summary}), 
their applicability is limited to offline cases where all the text collections are given in advance, 
thus requiring to rebuild indexes when additional texts are given. 
Evenworse, they need to store whole input texts in memory for constructing indexes, which requires a large amount of working space. 
The problem is especially serious when we process massive collections of highly repetitive texts, which is ubiquitous in the big data era. 
An open challenge is to develop an online self-indexed grammar compression not only with a small working space for a large input 
but also with a functionality of updating data structures for building self-indexes from new additional texts. 

\begin{table}[t] \label{tab:summary}
\begin{center}
{\footnotesize
  \caption{Comparison with offline methods. Construction time, search time and extraction time are presented
    in big $O$ notation that is omitted for space limitations.
    $N$ is the length of text, $m$ is the length of query pattern,
    $n$ is the number of variables in a grammar, $\sigma$ is alphabet size,
    $h$ is the height of the parse tree of the straight line program,
    $z$ is the number of phrases in LZ77, $d$ is the length of nesting in LZ77,
    $occ$ is the number of occurrences of query pattern in a text,
    $occ_q$ is the number of candidate appearances of query patterns,
    $\logstar$ is the iterated logarithm,
    and $\alpha \in (0,1]$ is a load factor for a hash table.
    $\lg$ stands for $\log _2$.}
}
\label{tbl:space}
{\footnotesize
  \begin{tabular}{l|c|c|c}
    &Working space(bits)&Index size(bits)&Algorithm\\ \hline  	
    LZ-index\cite{Navarro04JDA} &$O(N)$&$z\lg N+ 5z\lg N$ & Offline \\
    & & $ - z\lg z + o(N)+O(z)$ & \\ \hline
    Gagie et al.\cite{Gagie2012} & $O(N)$ & $2n\lg n +O(z\lg N$ & Offline \\
    & & $+z\lg z \lg\lg z)$ & \\ \hline
    SLP-index\cite{Claude2010,Claude12} & $O(N)$ & $n\lg N+ O(n\lg n)$ & Offline \\ \hline
    ESP-index\cite{Takabatake2014} & $O(N) $ & $n\lg N + n\lg n$ & Offline \\
    & & $+ 2n + o(n\lg n)$  & \\ \hline \hline
    OESP-index & $n\lg N + O((n+\sigma) \lg (n+\sigma))$ & $n\lg N + O((n+\sigma) \lg (n+\sigma))$ & Online 
\end{tabular}
}
\end{center}
\begin{center}
{\footnotesize
  \begin{tabular}{l|c|c|c}
                                 & Construction time & Search time & Extraction time \\ \hline
    LZ-index \cite{Navarro04JDA} & $N\lg \sigma$ & $m^2 d + (m + occ)\lg z$ & $md$ \\ \hline
    Gagie et al. \cite{Gagie2012} & $N$ & $m^2 + (m+occ)\lg\lg N$ & $ m + \lg \lg N$\\ \hline
    SLP-index \cite{Claude2010,Claude12} & $N$ & $m^2 + h(m+occ)\lg n$ & $(m + h)\lg n$\\ \hline
    ESP-index \cite{Takabatake2014} & $\frac{1}{\alpha}N\logstar N$ expected & 
$\lg\lg n(m+occ_q \lg m \lg N)\logstar N$ & $\lg\lg n(m+ \lg N)$\\ \hline \hline
    OESP-index & $\frac{1}{\alpha}N\lg (n+\sigma )\logstar N$ & 
$\lg (n+\sigma )(\frac{m}{\alpha}+occ_q(\lg N + \lg m\logstar N))$  & $\lg (n+\sigma )(m+ \lg N)$\\
    & expected & expected & \\
  \end{tabular}
}

\end{center}
\end{table}

Edit-sensitive parsing (ESP)~\cite{Cormode07} is an efficient parsing algorithm originally developed for
approximately computing edit distances with moves between texts. ESP builds
from a given text a parse tree that guarantees upper bounds of parsing discrepancies
between different appearances of the same subtext.
Maruyama et al.~\cite{ESP} presented a grammar-based self index called ESP-index on the notion of ESP and 
Takabatake et al.~\cite{Takabatake2014} improved ESP-index for fast query searches by using GMR's rank/select operations for general alphabet~\cite{Golynski06}. 
Unlike other grammar-based self-indexes, they perform top-down searches for finding candidate
appearances of a query text on the data structure by leveraging the upper bounds of parsing discrepancies in ESP. 
However, their applicability is limited to offline cases. 

In this paper, we present an online self-indexed grammar compression named OESP-index for building a self-index by reading input characters one-by-one. 
As far as we know, OESP-index is the first method for building grammar-based self-indexes in an online manner. 
OESP-index is built on the notion of ESP and its data structures are constructed by leveraging the idea behind fully-online LCA (FOLCA)~\cite{Maruyama2013,Maruyama2014}, 
an efficient online grammar compression that builds a context-free grammar (CFG) from an input text and encodes it into a succinct representation. 
We present a novel query search and random access algorithms for OESP-index and discuss their efficiency. 

Experiments were performed on retrieving query texts from a benchmark collection of highly repetitive texts. 
The performance comparison with other algorithms demonstrates OESP-index's superiority.

%% file: sec2.tex
\section{Preliminaries}
\subsection{Basic notations}
Let $\Sigma$ be a finite alphabet and $\sigma=|\Sigma|$.
The length of string $S$ is denoted by $|S|$.
The set of all strings over et $\Sigma$ is denoted by $\Sigma^*$.
The set of all strings of length $k$ is denoted by $\Sigma^k$.
We assume a recursively enumerable set ${\cal X}$ of variables with $\Sigma \cap {\cal X} = \emptyset$. 
$S[i]$ and $S[i,j]$ denote the $i$-th symbol of string $S$ and the substring from $S[i]$ to $S[j]$, respectively.
$\lg$ stands for $\log_2$.
Let $\lg^{(1)}u = \lg u$, $\lg^{(i+1)}u = \lg \lg^{(i)}u$, and $\logstar u = \min\{i\mid \lg^{(i)}u\leq 1\}$.
Practically $\logstar u=O(1)$ since $\logstar u\leq 5$ for $u\leq 2^{65536}$.

\subsection{Straight-line program (SLP)}
A context-free grammar (CFG) in Chomsky normal form is a quadruple $G=(\Sigma,V,D,X_s)$ where $V$ is a finite subset of ${\cal X}$,
$D$ is a finite subset of $V\times (V\cup \Sigma)^2$ and $X_s \in V$ is the start symbol.
An element in $D$ is called production rule.
A variable in $V$ is called nonterminal symbol.
$val(X_i)$ denotes the string derived from $X_i \in V$.
For $X_1,X_2,...,X_k \in V$, let $val(X_1,X_2,...,X_k)=val(X_1)val(X_2)...val(X_k)$.
A grammar compression of $S$ is a CFG that derives $S$ and only $S$.
The size of a CFG is the number of variables, i.e., $|V|$ and let $n=|V|$.

The parse tree of $G$ is a rooted ordered binary tree such that (i) internal nodes are labeled by
variables in $V$ and (ii) leaves are labeled by symbols in $\Sigma$, i.e., the label sequence in leaves is equal to input string $S$.
In a parse tree, any internal node $Z$ corresponds to a production rule $Z\to XY$ and has
a left child with label $X$ and a right child with label $Y$.
A partial parse tree~\cite{Rytter03} is an ordered tree formed by traversing the parsing tree in a depth-first manner and 
pruning out all descendants under every node of variables appearing no less than twice. 

{\em Straight-line program (SLP)}~\cite{SLP} is defined as a grammar compression over $\Sigma\cup V$
and its production rules are in the form of $X_k\to X_iX_j$ where $X_k,X_i,X_j \in \Sigma\cup V$ and
$1\leq i,j<k\leq n + \sigma$.

\subsection{Phrase dictionary and reverse dictionary}
A phrase dictionary is a data structure for directly accessing a digram $X_iX_j$ from a given $X_k$ if $X_k\to X_iX_j\in D$.
It is typically implemented by an array requiring $2n\log{(n+\sigma)}$ bits for storing $n$ production rules.
A reverse dictionary $D^{-1}$ is a mapping from a digram to an associated variable.
$D^{-1}(XY)$ returns the variable $Z$ if $Z\to XY \in D$; otherwise, 
it creates a new variable $Z' \notin V$ and returns $Z'$. 

\subsection{Succinct data structures}
We use the fully indexable dictionary (FID) for indexing bit strings.
Our method represents CFGs using a rank/select dictionary, a succinct data structure for a bit string $B$~\cite{Jacobson89} 
supporting the following queries: $\mbox{rank}_c(B,i)$ returns the number of occurrences of $c \in \{0,1\}$ in 
$B[0,i]$; $\mbox{select}_c(B,i)$ returns the position of the $i$-th occurrence of $c\in\{0,1\}$ in $B$; 
$\mbox{access}(B,i)$ returns $i$-th bit in $B$.
Data structures with $|B| + o(|B|)$ bit storage to achieve $O(1)$ time rank and select queries~\cite{Raman07} have been presented. 

For online grammar compression, we adopt the dynamic range min/max tree (DRMMT)~\cite{NavSad12} for online construction of parse tree.
We can obtain $parent(B,i)$, the parent of node $i$ of DRMMT $B$ in $O(\frac{\lg n}{\lg\lg n})$ time 
where $n$ is the number of nodes of the tree.
We consider the wavelet tree (WT)~\cite{Grossi03}, an extension of FID for general alphabet.
A WT is a data structure for a string over finite alphabets, and 
it can compute the rank and select queries on a string $S$ over $\Sigma^*$ in $O(\log{\sigma})$ time and using $|S|\log\sigma(1+o(1))$ bits.

\section{Edit Sensitive Parsing (ESP) and Fully-online LCA (FOLCA)}\label{sec:esp}
\begin{figure}[t]
\begin{center}
\includegraphics[width=0.9\textwidth]{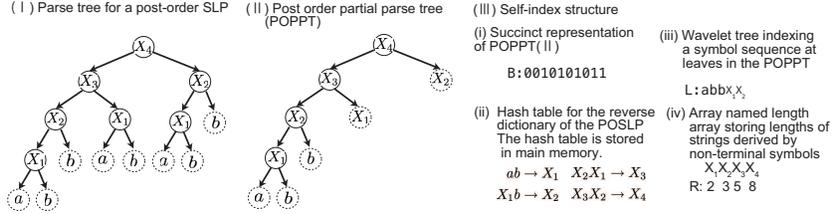}
\end{center}
\vspace{-0.5cm}
\caption{Example of parse tree, post order partial parse tree and self-index structure. The self-index structure consists of four 
data structures which are directly built from the parse tree.}
\label{fig:outline}
\end{figure}

We review the ESP algorithm~\cite{Cormode07} and its online variant named FOLCA~\cite{Maruyama2013} in this section. 
The original ESP is an offline algorithm and builds a parse tree named ESP-tree from a given string. 
ESP-trees are complete, balanced binary trees each subtree of which is 2-tree in the form of $X\to AB$ or 2-2-tree in the form of $X\to AY$ and $Y\to BC$. 
The algorithm partitions a string $S$ into non-overlapping substrings $S_1S_2\cdots S_\ell$ each of which 
belongs to one of three substring types. 
Type1 is a substring of a repeated symbol, i.e., $a^k$ for $a\in \Sigma$ and $k > 1$; 
type2 is a substring longer than $2\lg^*|S|$ not including type1 substrings; 
type3 is a substring that is neither type1 nor type2. 

The parsing algorithm parses each substring $S_i$ according to three substring types. 
For type1 and type3 substrings, the algorithm performs the left aligned parsing as follows. 
If $|S_i|$ is even, the algorithm builds a 2-tree from $S_i[2j-1,2j]$ for each $j \in \{1,2,...,|S_i|/2\}$;
Otherwise, the algorithm builds a 2-tree from $S_i[2j-1,2j]$ for each $j \in \{1,2,...,\lfloor(|S_i|-3)/2\rfloor\}$, and 
it builds a 2-2-tree from the last trigram $S_i[|S_i|-2,|S_i|]$. 
For type2 substrings, the algorithm further partitions substring $S_i$ into short substrings of length two or three by using 
an efficient string partitioning procedure named {\em alphabet reduction}~\cite{Cormode07}, and 
it builds 2-trees for substrings of length two and 2-2-trees for substrings of length three. 
The parsing algorithm generates a new shorter string $S^\prime$ of length from $|S|/3$ to $|S|/2$, and 
it parses $S^\prime$. 
The above process iterates on the new sequence until its length is one. 

FOLCA is an online algorithm that builds ESP-tree as {\em a post-order partial parse tree (POPPT)} using the parsing rule in the ESP algorithm  
from a given string in an online manner. 
A POPPT and the {\em post-order SLP (POSLP)} corresponding to a POPPT are defined as follows. 
\begin{definition}[POPPT and POSLP~\cite{Maruyama2013}]
A POPPT is a partial parse tree whose internal nodes have post-order variables. 
A POSLP is an SLP whose partial parse tree is a POPPT. 
\end{definition}
Figure~\ref{fig:outline}-(I) and -(II) show an example of parse tree and POPPT. 

Since FOLCA builds POPPT using the rules in the ESP algorithm, it can exploit advantages existing in both SLP and ESP-tree. 
Given a string $S$, FOLCA builds the POPPT of height $O(\lg|S|)$ in $O(|S|\lg^*|S|)$ time. 
FOLCA's worst-case approximation ratio to the smallest CFG is $O(\lg^*|S|)$. 
OESP-index directly encodes FOLCA'S POPPT into a succinct representation and build an index structure in an online manner for fast query searches and 
substring extractions, which is explaned in the next section. 

%% file: sec3.tex
\section{Index structure of OESP-index}

%
\begin{figure}[t]
\begin{center}
\includegraphics[width=0.9\textwidth]{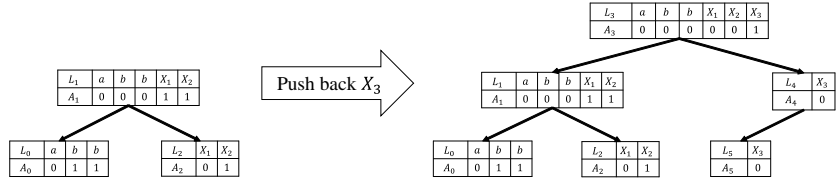}
\end{center}
\vspace{-0.6cm}
\caption{Example of dynamic wavelet tree.
  $L$ is the leaf label of POPPT;
  Code is the integer representation of $L$;
  $B_i$ is the bit vector representing elements in code; Only $A_i$ at each node is stored.}
\label{fig:dwt}
\end{figure}

OESP-index's succinct representation consists of four data structures: 
(i) $B: $ succinct tree of POPPT, (ii) $H: $ hash table (iii) $L: $ non-negative integer array indexed by wavelet tree  
and (iv) $R: $ non-negative integer array.

Succinct tree of POPPT $B$ is a bit string  made by traversing POPPT in post-order, and 
putting '0' if a node is a leaf and '1' otherwise. 
The last bit '1' in $B$ represents a virtual node and $B$ is indexed by the DRMMT~\cite{NavSad12}.
The succinct tree supports the following three operations: 
$parent(B,i)$ return the parent node of a node $i$;
$left\_child(B,i)$ returns the left child of a node $i$; 
$right\_child(B,i)$ returns the right child of a node $i$; 
They are computed in $O(\lg{n}/(\lg\lg{n}))$ time.
The space for our succinct tree is at most $2n+o(n)$ bits. 

A reverse dictionary $H:(V\cup \Sigma)\times (V\cup \Sigma)\to V$ is implemented by a chaining hash table. 
Let $\alpha$ be a constant called a load factor. 
The hash table has $\alpha n$ entries and each entry stores a list of 
integers $i$ representing the left hand side of a rule $X_i\to X_jX_k$. 
The size of the data structure is $\alpha n \lg(n+\sigma)$ bits for the hash table 
and $n\lg{(n+\sigma)}$ bits for the lists. 
Thus, the total size is $n(1+\alpha)\lg{(n+\sigma))}$ bits. 
The access time is expected $O(1/\alpha)$ time. 

Non-negative integer array $L$ stores symbols at leaves from the leftmost leaf to the rightmost leaf in the POPPT. 
$L$ is indexed by {\em dynamic wavelet tree (DWT)} that is presented in the next subsection. 
Each element of non-negative integer array $R$ is the length of the string derived from a variable, i.e., $|val(X_i)|$ for $X_i \in V$. 
The size of $R$ is $n\lg{|S|}$ bits. 
Figure~\ref{fig:outline}-(III)-(iv) shows an example of data structures. 

\subsection{Dynamic wavelet tree (DWT)}
Our DWT is a wavelet tree supporting a operation of adding an element to the tail of a sequence.
Such a operation is called a pushback that is necessary for implementing DWT. 
A wavelet tree for sequence $L$ over range of alphabet and variables $[1..(n+\sigma)]$ can be recursively described over 
sub-range $[a..b] \subseteq [1..(n+\sigma)]$. 
A wavelet tree over range $[a..b]$ is a binary balanced tree with $b-a+1$ leaves. 
If $a=b$, the tree is just a leaf labeled $a$. 
Otherwise it has an internal root node that represents $L$. 
The root has a bitstring $A_{root}[1,|S|]$ defined as follows: 
if $L[i] \leq (a+b)/2$ then $A_{root}[i]=0$, else $A_{root}[i]=1$. 
We define $L_0[1,\ell_0]$ as the subsequence of $L$ formed by the symbols $c\leq (a+b)/2$, and 
$L_1[1,\ell_1]$ as the subsequence of $L$ formed by the symbols $c>(a+b)/2$. 
Then, the left child of the root is a wavelet tree for $L_0[1,\ell_0]$ over range $[a..\lfloor (a+b)/2 \rfloor]$ 
and the right child of the root is a wavelet tree for $L_1[1,\ell_1]$ over range $[1+\lfloor (a+b)/2 \rfloor ..b]$.

Implementing WTs without pointers uses a small space of $n\lg(n+\sigma)+o(n\lg(n+\sigma))$ bits, but 
supporting the pushback operation is difficult. 
Thus, we implement DWTs using pointers where the binary tree is explicitly represented. 
When a new symbol exceeding the representation ability of the current binary tree in DWT is added to DWT, 
DWT adds new nodes to the binary tree, resulting in increasing the height of the tree. 
The space of DWT uses $(3n+2\sigma)\lg(n+\sigma)+o(n\lg(n+\sigma))$ bits.
Figure~\ref{fig:dwt} shows an example of DWT.

%
\subsection{Complexity for building OESP-index}
\begin{theorem}~\label{construction-time}
The size of OESP-index is $n\lg{|S|}+O((n+\sigma)\lg (n+\sigma))$ bits.
The construction time is $O(\frac{1}{\alpha}|S|\lg (n+\sigma)\logstar |S|)$ and the memory consumption is the same as the index size,
where $S$ is an input string, $n$ is the number of variables, $\alpha$ is a load factor of the hash table,
and we assume the size of alphabet is constant.
The update time for the next input symbol is $O(\frac{1}{\alpha}\lg (n+\sigma)\logstar |S|)$.
\end{theorem}
\begin{proof}
The size of the length array $R$ is $n\lg{|S|}$ bits for $n$ variables.
The size of $B$ is $2n+o(n)$ bits and the size of $L$ and $H$ are $O((n+\sigma)\lg (n+\sigma))$ bits each.
We can access $Z=H(XY)$ in $O(1/\alpha)$ time for a load factor $\alpha\in (0,1]$.
The alphabet reduction is iterated at most $\lg^*{|S|}$ times for each symbol.
The time to get the parent and left/right children of a node in the partial parse tree
is $O(\lg (n+\sigma))$ using the rank/select over the DWT for $L$.
Thus, the construction time of the parse tree is $O(\frac{1}{\alpha}|S|\lg (n+\sigma)\logstar{|S|})$.
Analogously, the update time is clear.
\end{proof}

\subsection{Query search and substring extraction}

\begin{algorithm}[t]
  {\scriptsize
    \caption{{\sc NextCore} on implicit parse tree POPPT$(B,L)$}
    \label{algo:core}
  }
  {\scriptsize
    \begin{algorithmic}[1]
      \State $v$: the leftmost occurrence node of maximal core, $p$: empty stack 
      \Function{{\sc NextCore}}{$v$, $p$} 
      \If{$v \neq$ root}
      \If{$v$ is the left child of $parent(B,v)$} \Comment 
      \State $p.push(left)$
      \Else
      \State $p.push(right)$
      \EndIf
      \State {\sc NextCore}($parent(B,v)$, $p$)
      \State $i\leftarrow 1$
      \State $p.pop()$
      \While{($u$ = ${\rm select}_v(L,i)$) $\neq$ NULL} \Comment  $(u, p)$: the next occurrence on explicit tree
      \If{$u$ is left child of $parent(B, u)$}
      \State $p.push(left)$
      \Else
      \State $p.push(right)$
      \EndIf
      \State {\sc NextCore}($parent(B,u)$, $p$)
      \State $p.pop()$
      \State $i \leftarrow i+1$
      \EndWhile
      \EndIf
      \EndFunction
    \end{algorithmic}
  }
\end{algorithm}

For a node $v$ of a parse tree of the string $S\in\Sigma^*$, 
and $yield(v_1\cdots v_k)=yield(v_1)\cdots yield(v_k)$.
$Label(v)$ denotes the label of $v$ and $Label(v_1\cdots v_k)=Label(v_1)\cdots Label(v_k)$.
If $Label(v)=X$, $yield(X)$ is identical to $yield(v)$.
$lca(u,v)$ is the lowest common ancestor of $u,v$.
For a pattern $P\in \Sigma^*$, nodes $\{v_1,\ldots,v_k\}$ such that $yield(v_1\cdots v_k)=P$ are called {\em embedding nodes} of $P$. 
For embedding nodes $\{v_1,\ldots,v_k\}$, string $Q=Label(v_1\cdots v_k)$ is called an evidence of pattern $P$.
Since the trivial evidence $Q$ identical to $P$ always exists, the notion of evidence is well-defined.
In addition, for embedding nodes $\{v_1,\ldots,v_k\}$, a node $z$ such that $z=lca(v_1,v_k)$ is called an occurrence node of $P$


The next theorem tells that we can find shorter evidence depending on $|P|$.
\begin{theorem}(~\cite{ESP})
There exists an evidence $Q=Q_1\cdots Q_t$ of $P$ such that 
each $Q_i$ is a maximal repetition or a symbol and $t=O(\lg |P|\logstar |S|)$. 
\end{theorem}

The time to find the evidence $Q$ of pattern $P$ is bounded by the construction time of 
the parsing tree of $P$.
In our data structure of OESP, the time to find the evidence $Q$ is estimated as follows.

\begin{theorem}~\label{evidence-time}
The time to find $Q$ is $O(\frac{1}{\alpha}|P|\lg (n+\sigma)\logstar{|S|})$.
\end{theorem}
\begin{proof}
This bound is clear by Theorem~\ref{construction-time}.
\end{proof}

Let us consider the simple case that $|Q_i|=1$ for any $i$.
In this case, $Q$ contains no repetition such that $Q=q_1\cdots q_t\in \Sigma^t$.
A symbol $q_k$ is called a maximal core if $|yield(q_k)|\geq |yield(q_i)|$ for any $i$.
For an internal node $v$ of the parse tree $T$ of $S$ with $Label(v)=q_k$,
an ancestor $z$ of $v$ is the occurrence node of $P$ iff
all $q_1,\ldots,q_{k-1}$ and $q_{k+1},\ldots,q_t$ can be embedded around $v$.
Moreover, any occurrence node of $P$ is restricted by the case $Label(v)=q_k$.
For the general case $Q_i= a^\ell$ $(\ell\geq 2)$, i.e., $Q_i$ is a repetition,
we can reduce the embedding of $a^\ell$ to the embedding of a string $AB\cdots C$ 
of length at most $O(\lg \ell)$ such that $yield(AB\cdots C)=a^\ell$.
Thus, the embedding of type1 string is easier than others, and then,
without loss of generality, we can assume $|Q_i|=1$ for any $i$.

The remaining task of the search problem is the random access to all occurrences of 
the maximal core $q_k$ over the POPPT, the pruned parse tree.
By the definition of POPPT, the internal node with rank $k$ is 
the leftmost occurrence of the symbol $q_k$ itself.
In the previous indexes~\cite{ESP,Takabatake2014}, a next occurrence of $q_k$ is obtained using
a data structure based on the renaming variables in a lexicographic order.
This data structure is not, however, dynamically constructable.
Therefore, we develop the search algorithm {\sc NextCore} (Algorithm~\ref{algo:core}) for the OESP-index.

The {\sc NextCore} visits all occurrences of the maximal core on the parse tree $T$ using its implicit POPPT $T'$.
When {\sc NextCore} receives a candidate node $v$ containing a maximal core $q$ as its descendant,
it computes the pair $(u,p)$ where $u$ is the next occurrence of $v$ in $T'$ and $p$ is the path from $u$ to $q$. 
Thus, $(u,p)$ indicates the occurrence of $q$ in the explicit parse tree.
We show the correctness of this algorithm and its complexity.

\begin{lemma}~\label{next-core}
The {\sc Nextcore} find any occurrence of the maximal core exactly once.
The amortized time to find a next occurrence is $O(\frac{\lg n\lg{|S|}}{\lg\lg n})$.
\end{lemma}
\begin{proof}
Let $T$ be the parse tree and $T'$ be the POPPT $(B,L)$.
By the definition of $T'$, any internal node $x$ of $B$ is the variable itself, i.e., $Label(x)=x$.
For the maximal core $q$, let $v_1>v_2>\cdots >v_k$ be the post-order of its occurrences in $T$.
We show that the algorithm finds any $v_i$ as $(u,p)$ by induction on $i$.
Given $q$, the internal node $q$ of $B$ represents the leftmost occurrence of $q$ itself.
Then, for the base case $i=1$, the occurrence is obtained $v_1$ as $(q,p)$ with $|p|=0$.
Assume the induction hypothesis on some $i$.
Since the node $v_{i+1}$ was pruned in $T'$, let $u$ be the leaf of $T'$ 
corresponding to the root of the pruned maximal subtree containing $v_{i+1}$.
For $Label(u)=u'$, there is the leftmost occurrence of $u'$ as an internal node of $B$.
The subtree on the node $u'$ contains an occurrence of $q$ because
the two subtrees on $u$ and $u'$ in $T$ are identical each other.
Let $p$ be the path from $u'$ to $v'$ for some $v'\in\{v_1,\ldots,v_i\}$.
By the induction hypothesis, the algorithm finds $v'$ as $(u',p)$.
Then, $v_{i+1}$ can be also found as $(u,p)$.
On the other hand, any $(u,p)$ is unique, then the algorithm finds any occurrence of $q$ exactly once.
For the time complexity, the number of executed select operations is bounded by the number of different $(u,p)$
that is $O(occ_q \log{|S|})$ where $occ_q$ is the number of occurrences of $q$.
Each select operation on $L$ and parent operation on $B$ take $O(\lg (n+\sigma))$ and $O(\frac{\lg n}{\lg\lg n})$ time, respectively.
Therefore, the total time is $O(occ_q (\lg (n+\sigma) + \frac{\lg n\lg{|S|}}{\lg\lg n}))=O(occ_q (\frac{\lg n\lg{|S|}}{\lg\lg n}))$ and 
the amortized time to find a next occurrence of $q$ is $O(\frac{\lg n\lg{|S|}}{\lg\lg n})$.
\end{proof} 

\begin{theorem}~\label{search-time}
The counting/locating time of pattern and extraction time are 
$O(\lg (n+\sigma)(\frac{|P|}{\alpha}+occ_q(\lg{|S|}  + \lg{|P|}\logstar{|S|})))$ and 
$O(\lg (n+\sigma)(|P|+\lg |S|))$, respectively, where $P$ is a query pattern and 
$occ_q$ is the number of occurrences of the maximal core of $P$ in the parse tree.
\end{theorem}
\begin{proof}
Since we can get the length of the substring encoded by any variable in $O(1)$ time,
the locating time is same as the counting time.
Given the pattern $P$, as previously shown, the evidence $Q$ of $P$ is found in $O(\frac{1}{\alpha}|P|\lg (n+\sigma)\logstar{|S|})$ time.
For each occurrence of a maximal core, we can check if the sequence of symbols of length $O(\lg{|P|}\lg^*|S|)$ is embedded 
around the core in $O(\lg (n+\sigma)(\lg{|S|} + \lg |P|\lg^*{|S|}))$ time.
Therefore, by Lemma~\ref{next-core}, the total counting time of pattern is
\begin{eqnarray*}
  &O&(\frac{|P|}{\alpha}\lg (n+\sigma)\logstar{|S|} + occ_q \lg (n+\sigma)(\lg{|S|}+\lg{|P|}\logstar{|S|}) + occ_q \frac{\lg n\lg |S|}{\lg\lg n})\\
&=& O(\lg (n+\sigma)(\frac{|P|}{\alpha}+occ_q(\lg |S| + \lg |P|\logstar |S|))).
\end{eqnarray*}
On the other hand, for any $S[i,j]$ of length $m$,
we can find $S[i]$ in $O(\lg{|S|})$ time and 
visit all leaves in $S[i,j]$ in $O(|P|)$ time because 
the parsing tree is balanced.
This follows the extraction time.
\end{proof}

%% file: sec4.tex
\begin{figure}
  \begin{minipage}{0.5\hsize}
    \begin{center}
      \includegraphics[width=1.0\textwidth]{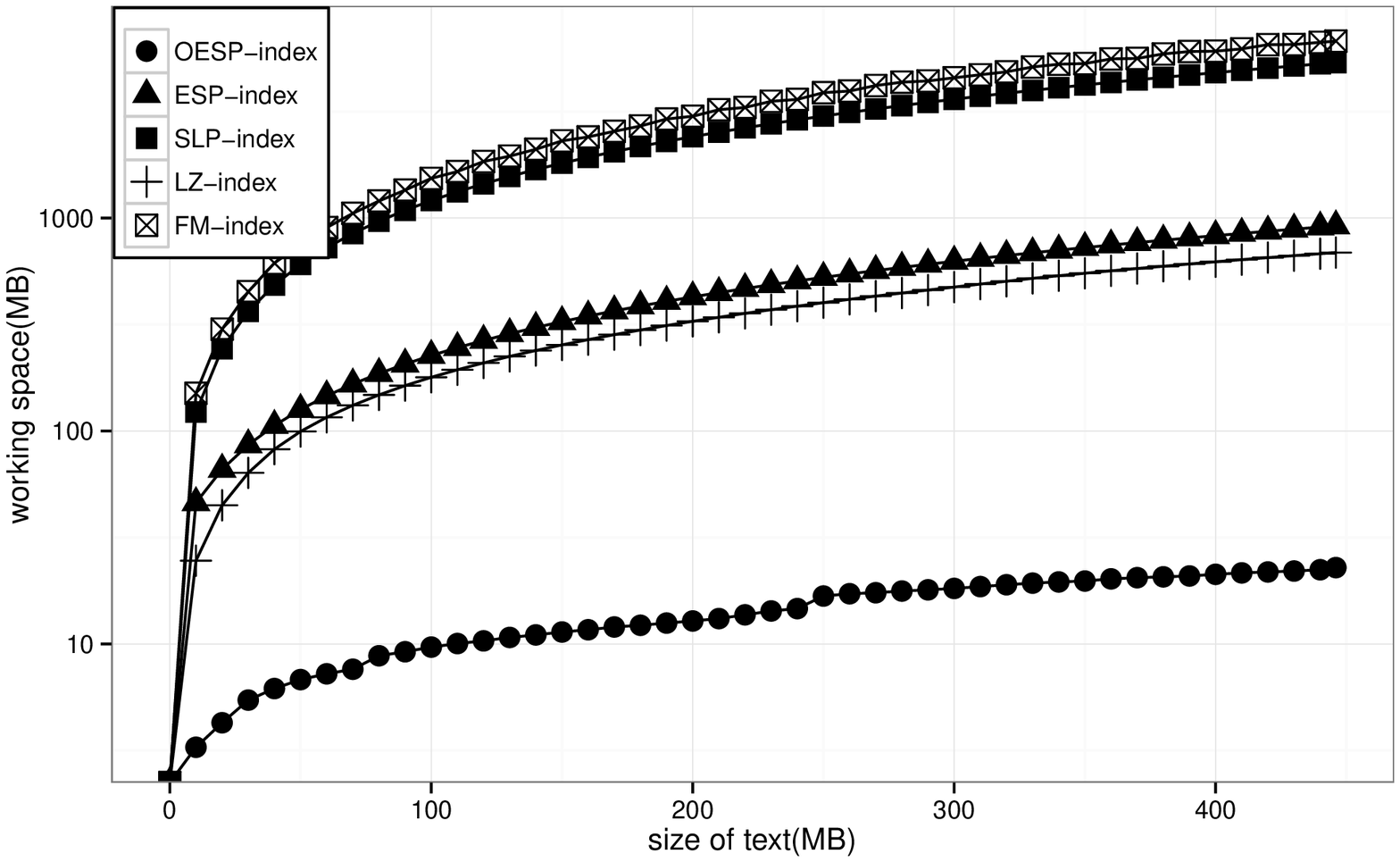}
    \end{center}
  \end{minipage}
  \begin{minipage}{0.5\hsize}
    \begin{center}
      \includegraphics[width=1.0\textwidth]{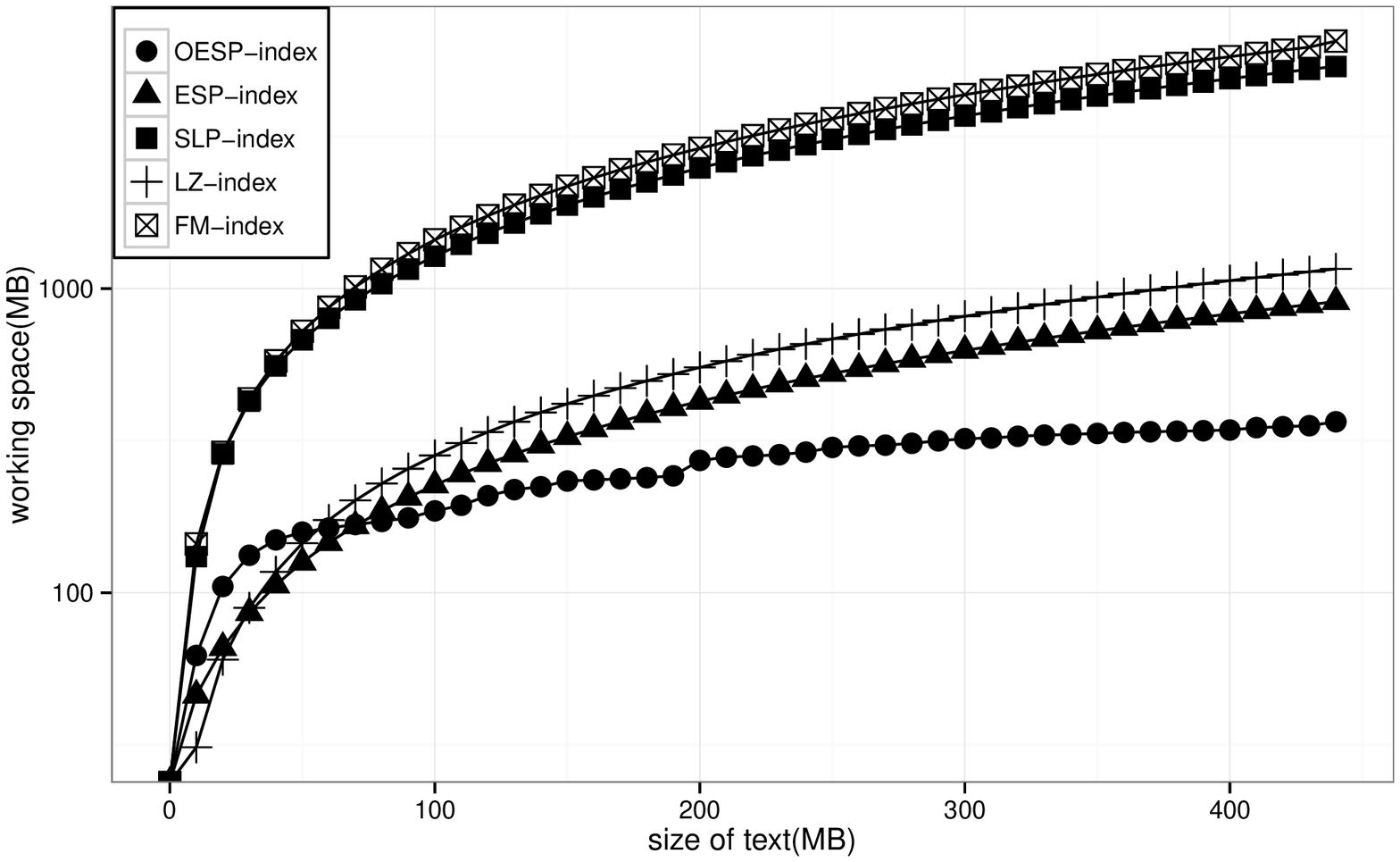}
    \end{center}
  \end{minipage}
  \caption{Working memory of each method in megabytes for einstein(left) and cere(right). }
   \label{fig:work}
  \begin{minipage}{0.5\hsize}
    \begin{center}
      \includegraphics[width=1.0\textwidth]{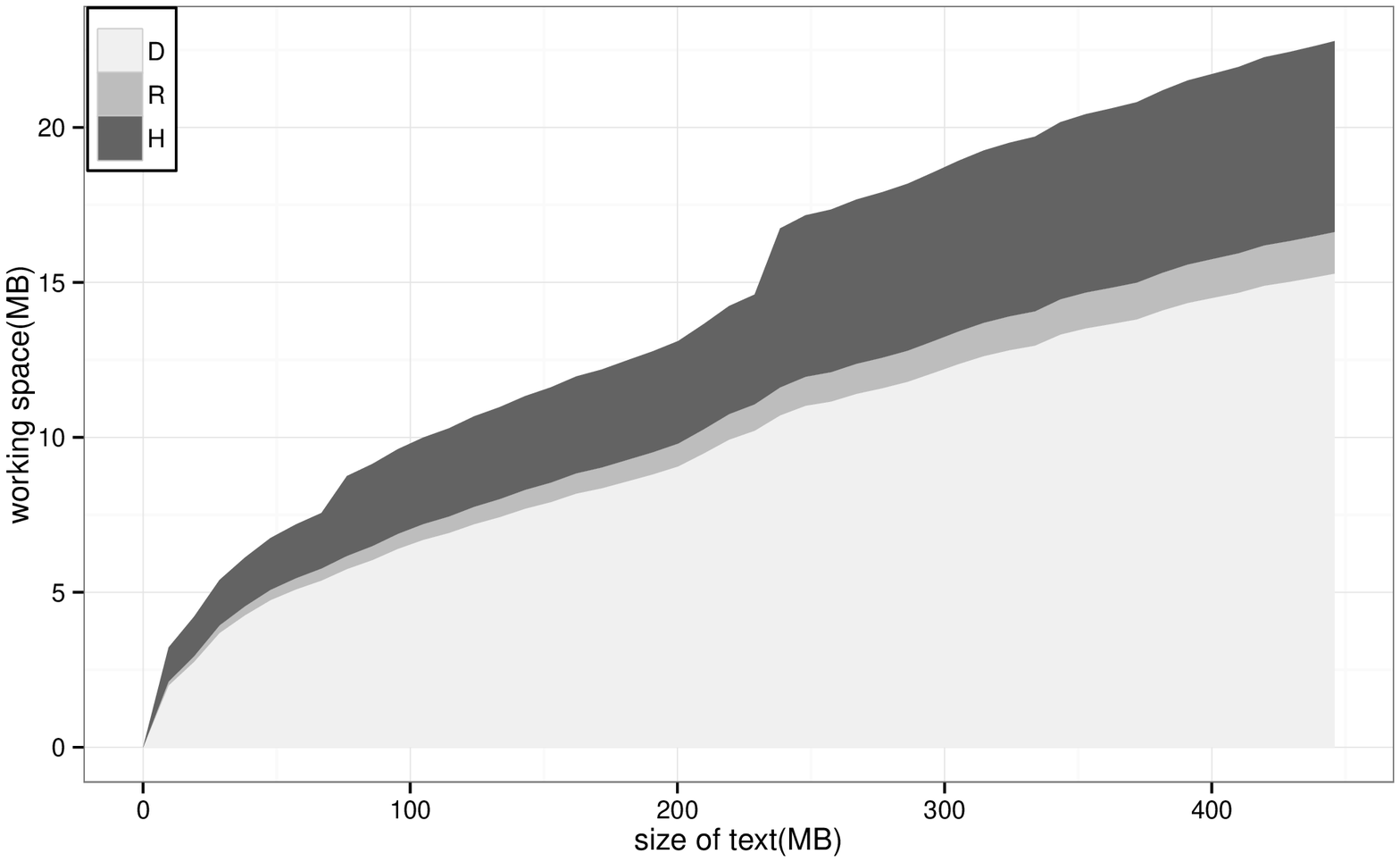}
    \end{center}

  \end{minipage}
  \begin{minipage}{0.5\hsize}
    \begin{center}
      \includegraphics[width=1.0\textwidth]{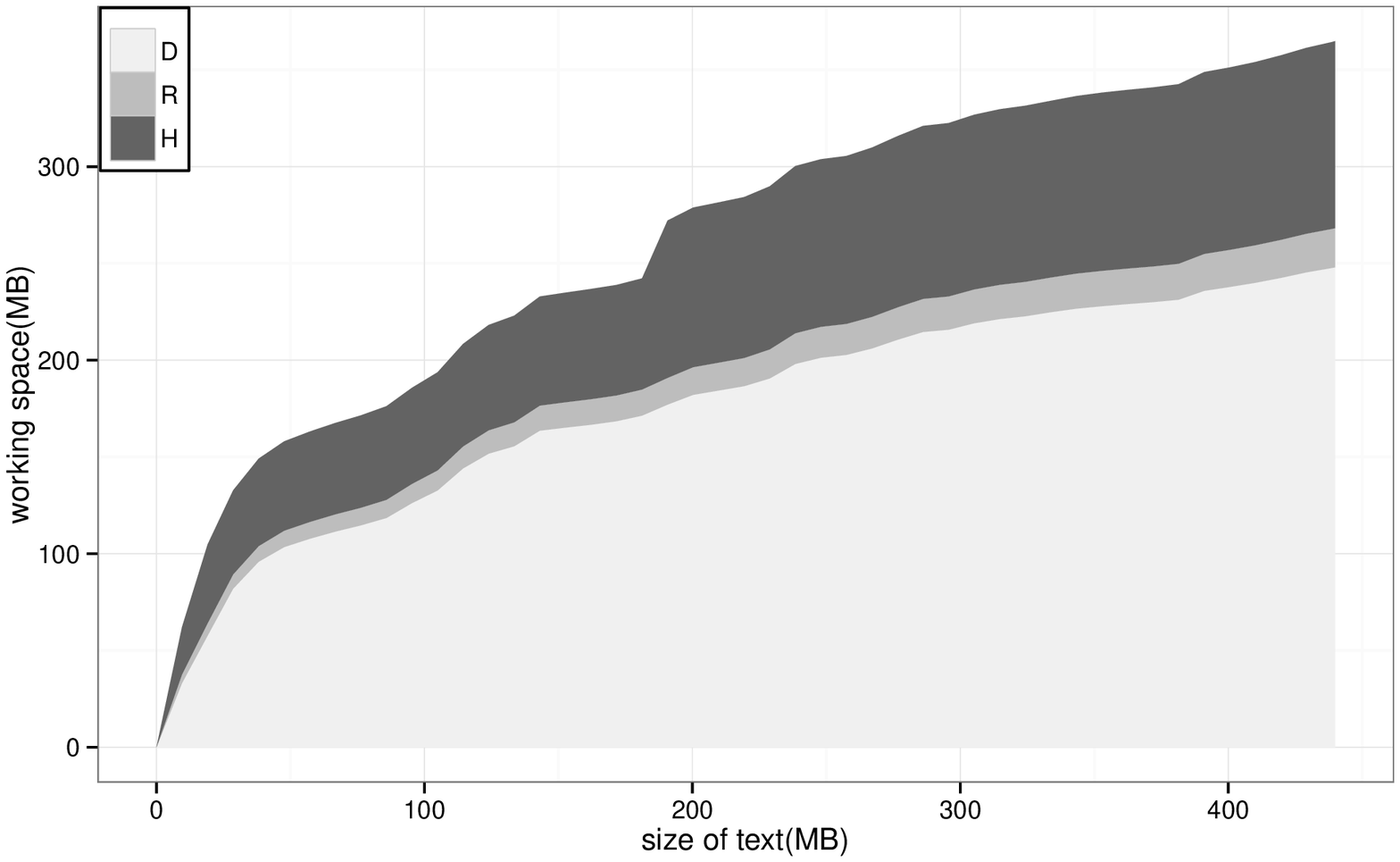}
    \end{center}

  \end{minipage}
  \caption{Working space of dictionary $D$, length array $R$ and hash table $H$ for einstein~(left) and cere~(right).}
   \label{fig:workp}

  \begin{minipage}{0.5\hsize}
    \begin{center}
      \includegraphics[width=1.0\textwidth]{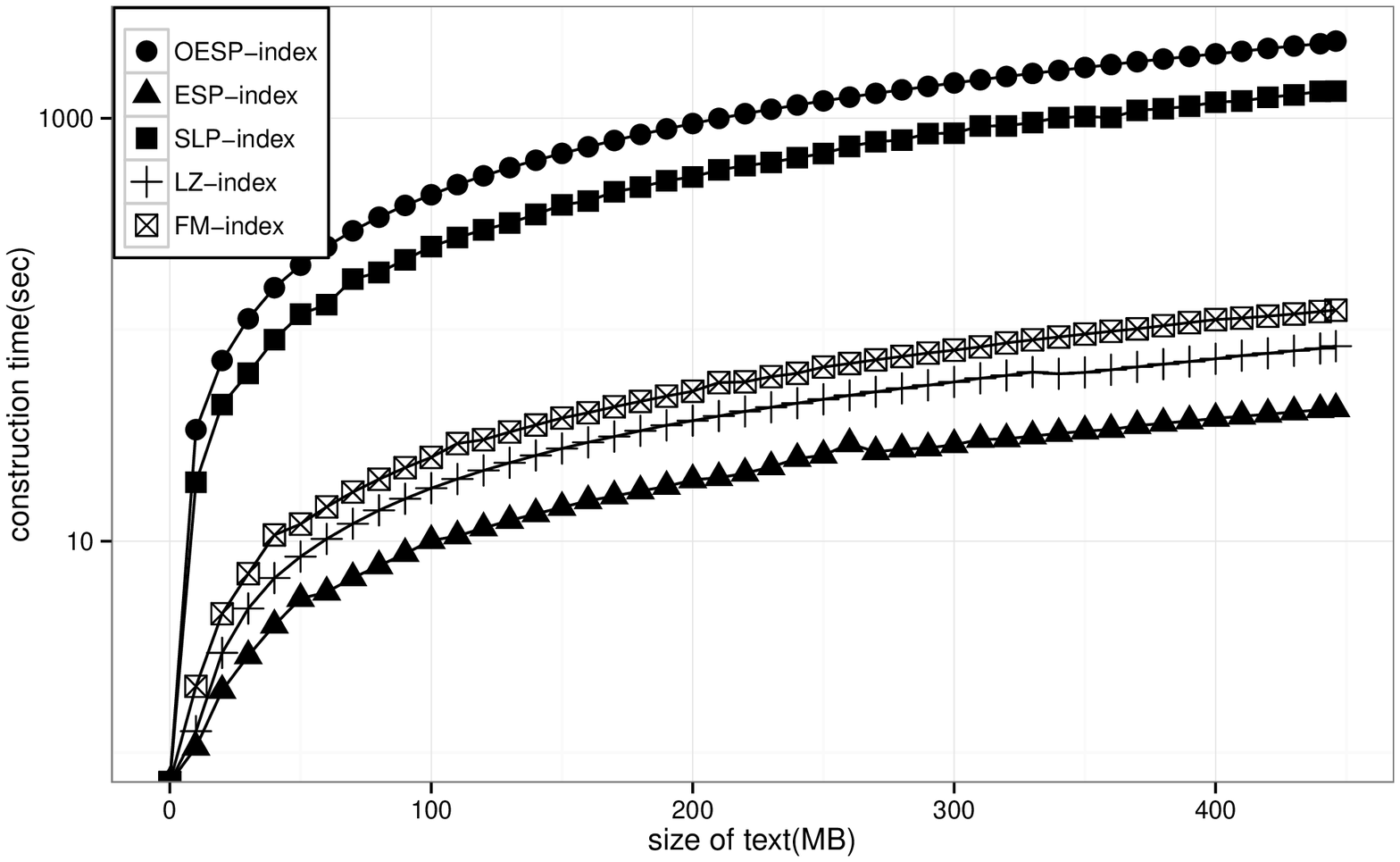}
    \end{center}
  \end{minipage}
  \begin{minipage}{0.5\hsize}
    \begin{center}
      \includegraphics[width=1.0\textwidth]{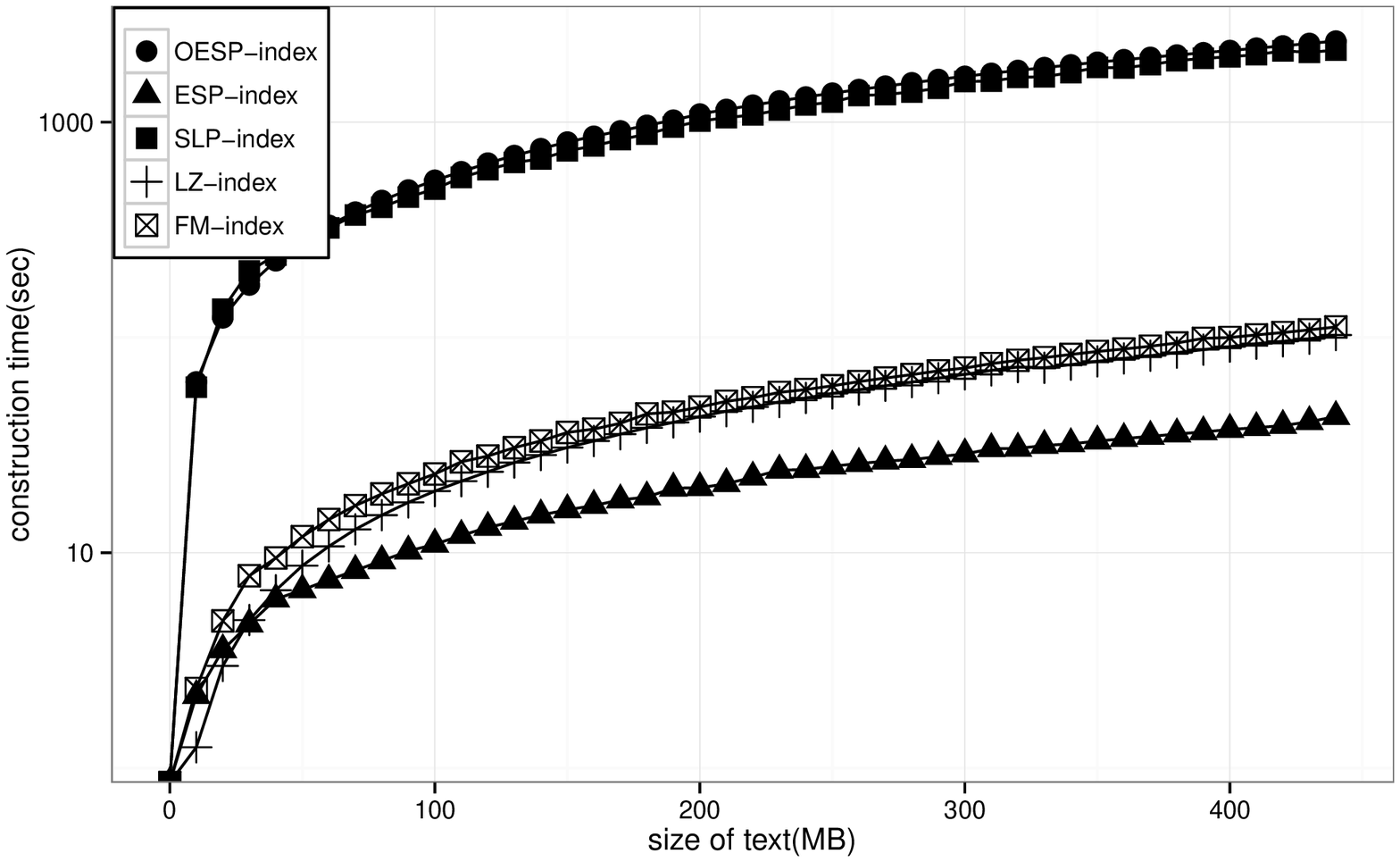}
    \end{center}

  \end{minipage}
  \caption{Construction time of each method in seconds for einstein(left) and cere(right). }
   \label{fig:con}

  \begin{minipage}{0.5\hsize}
    \begin{center}
      \includegraphics[width=1.0\textwidth]{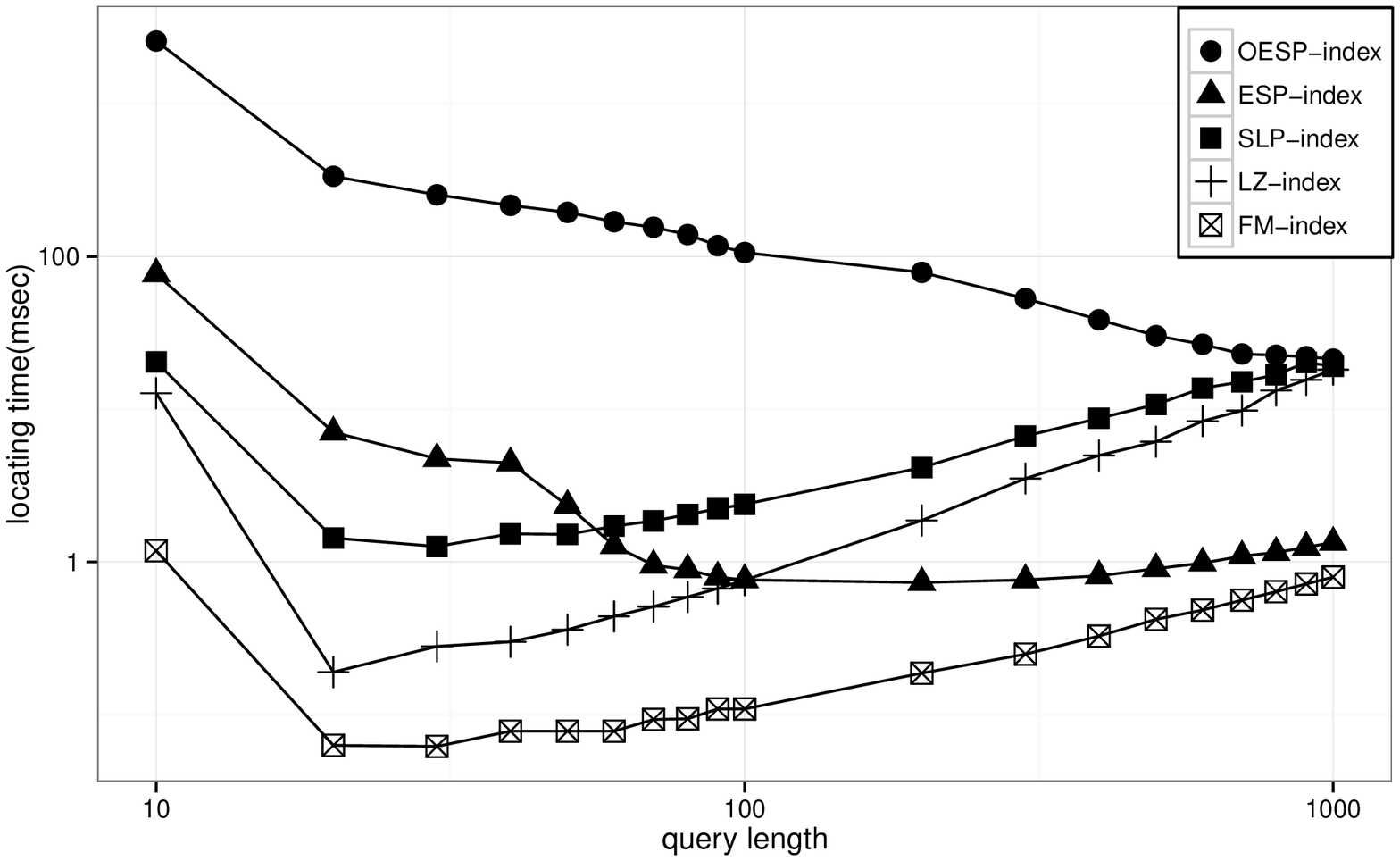}
    \end{center}
  \end{minipage}
  \begin{minipage}{0.5\hsize}
    \begin{center}
      \includegraphics[width=1.0\textwidth]{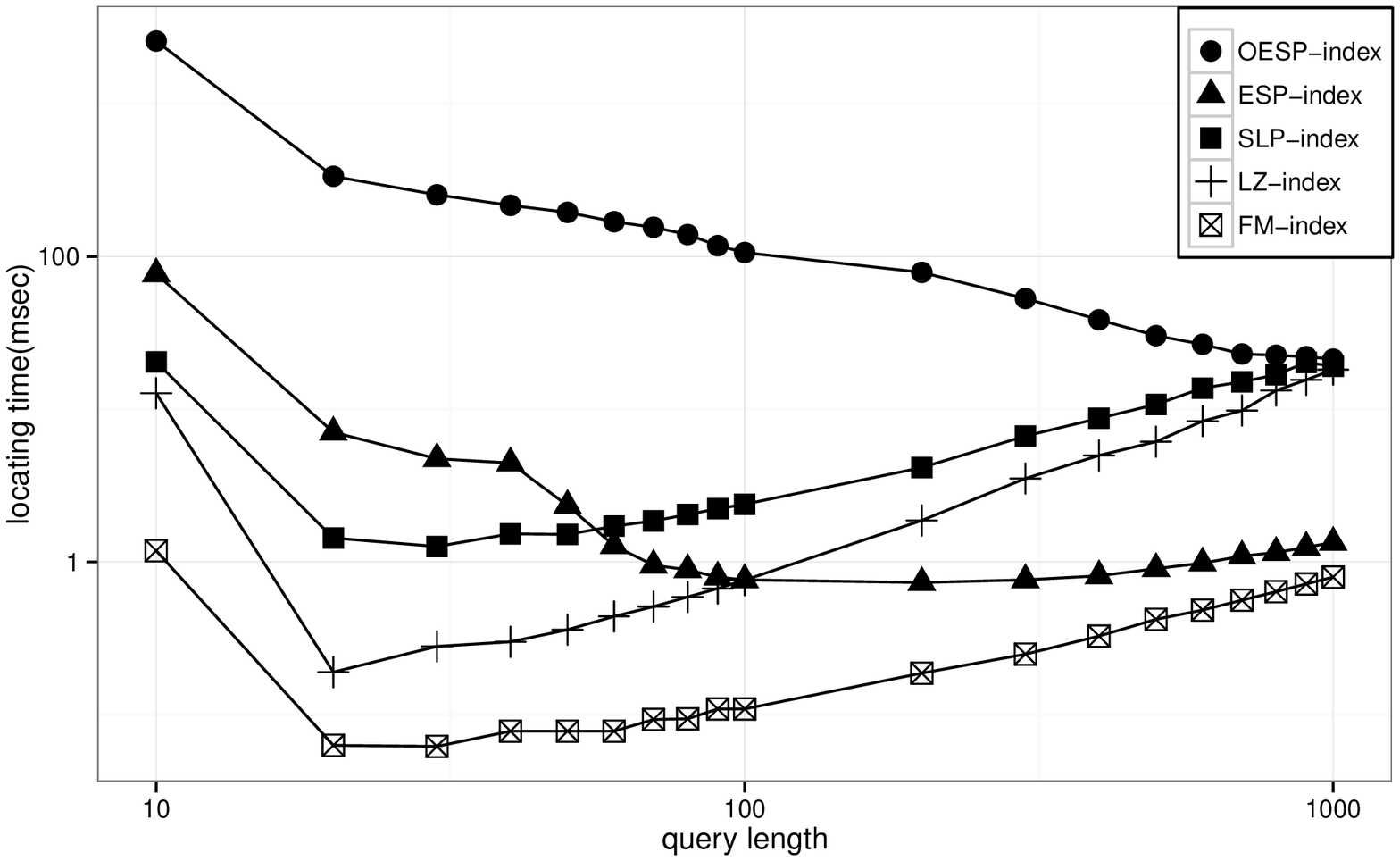}
    \end{center}
  \end{minipage}
  \caption{Locating time of each method in milliseconds for einstein(left) and cere(right). }
  \label{fig:locate}
  \end{figure}

\section{Experiments}
\begin{table}[tb]
  \caption{Index size in mega bytes(MB).}
  \begin{center}
  \begin{tabular}{l|c|c|c|c|c}
    & OESP-index &ESP-index & SLP-index & \
    LZ-index & FM-index \\ \hline 
    einstein  &22.84 & 1.76 & 2.28 & 177.02 & 942.85\\
    cere &364.92 & 27.40 & 45.74 &  438.05 & 806.52\\
  \end{tabular}
  \end{center}
  \label{tab:index}
  \vspace{0cm}
\end{table}

\begin{table}[tb]
  \caption{Working memory of dictionary $D$ consisting of the bit string $B$ and the dynamic wavelet tree $L$ for einstein and cere.}
  \begin{center}
  \begin{tabular}{l|c|c|c|c|c|c|c|c|c|c}
    \multicolumn{2}{l|}{Size of text(MB)} & $50$ & $100$ & $150$ & $200$ & $250$ & $300$ & $350$ & $400$ & all\\ \hline
    einstein & $B$(MB) & $0.04$ & $0.07$ & $0.07$ & $0.07$ & $0.14$ & $0.14$ & $0.14$ & $0.14$ & $0.14$\\ 
             & $L$(MB)& $5.06$ & $6.63$ & $7.84$ & $8.99$ & $10.88$ & $12.23$ & $13.38$ & $14.37$ & $15.14$ \\ \hline
     cere & $B$(MB) & $1.10$ & $1.10$ & $1.10$ &  $2.20$ &  $2.20$ &  $2.20$ & $2.20$ & $2.20$ & $2.20$ \\ 
            & $L$(MB) & $102.34$ &  $131.58$ & $164.05$ & $179.90$  & $199.08$ & $216.87$ & $225.70$ & $235.62$ & $245.81$\\
  \end{tabular}
  \end{center}
  \label{tab:dict}
  \vspace{0cm}
\end{table}

We evaluated the actual performance of OESP-index for real data\footnote{\url{http://pizzachili.dcc.uchile.cl/repcorpus/real/}}.
The environment is Intel(R) Core(TM)i$7$-$2620$M CPU($2.7$GHz) machine with $16$GB memory. 
We use einstein.en.txt (einstein, $446$ MB) and cere (cere, $440$ MB), where einstein is highly repetitive.

Compared self-indexes are offline version of ESP-index (ESP-index)\cite{Takabatake2014},
other grammar-based self-index (SLP-index)\cite{Claude2010,Claude12},
LZ-based index (LZ-index)\footnote{\url{http://pizzachili.dcc.uchile.cl/indexes/LZ-index/LZ-index1}}, and
BWT-based self-index (FM-index)\footnote{\url{https://code.google.com/p/fmindex-plus-plus/}}.
Figure~\ref{fig:work} shows the required working memory (MB) in response to an increase of input string.
For the offline algorithms, the working memory is evaluated for each static data with the indicated size.
Figure~\ref{fig:workp} is the breakdown of required memory by the data structures of OESP-index:
dictionary $D$, length array $R$, and hash table $H$. 
Besides, Table~\ref{tab:dict} is the breakdown of $D$ by the bit string $B$ and the wavelet tree $L$.

Table~\ref{tab:index} shows the size of indexes of all methods.
The size of OESP-index is smaller than LZ-index and FM-index but larger than ESP-index and SLP-index.
The increase of index size arise from DWT. 
Reducing this data size is an important future work.

The memory consumption of OESP-index is smallest for both type of data.
The required memory of OESP-index is 2.5\%(einstein) and 40\%(cere) of offline ESP-index.
The space efficiency of OESP-index comes down when the data is not large and not highly repetitive (Figure \ref{fig:workp} (right)).
Especially, $L$ represented by the dynamic wavelet tree (DWT) consumes a large space (Table \ref{tab:dict})
arising from the pointer and the reservation space of bit string in DWT.


Figure~\ref{fig:con} shows the construction time.
OESP-index is slowest for both data in all methods.
OESP-index is $57.1$ times (einstein) and $58.1$ times (cere) slower than ESP-index
because the original one can use GMR~\cite{Golynski06}, a faster wavelet tree algorithm but not available in the online version.

Figure~\ref{fig:locate} shows the search time.
Here the search time means the locating time since the counting time is almost same to the locating time.
We note that the result of SLP-index is not shown because it could not work for this data.
The range of the length of query pattern is $[10,1000]$.
The locating time of OESP-index is slowest in both data in all query length.
OESP-index is $163.2$ times (cere) and $24.9$ times (einstein) slower than ESP-index.

%% file: sec5.tex
\section{Conclusion}
We have presented OESP-index, an online self-indexed grammar compression. 
OESP-index is the first method for building grammar-based self-indexes in an online manner. 
Experimental results demonstrated OESP-index's potential for processing a large collection of highly repetitive texts. 
Future work is to make OESP-index scalable to massive collections of the same type, 
which is required in the big data era.